\documentclass[12pt]{article}

\usepackage{amssymb,amsmath,amsthm}
\usepackage{palatino}
\usepackage{slashbox}
\usepackage{ctable}
\usepackage{enumerate}

\newtheorem{thm}{Theorem}[section]
\newtheorem{defn}[thm]{Definition}
\newtheorem{lemma}[thm]{Lemma}

\newtheorem{cor}[thm]{Corollary}

\newcommand{\OA}{$\mathit{OA}$ }
\newcommand{\NA}{$\mathit{NA}$ }

\title{On Sisterhood in the Gale-Shapley Matching Algorithm
\author {
Yannai A. Gonczarowski
\thanks{Einstein Institute of Mathematics,
 Hebrew University, Jerusalem, Israel.
Email: \mbox{yannai@gonch.name}}
\ and Ehud Friedgut
\thanks{Einstein Institute of Mathematics,
Hebrew University, Jerusalem, Israel.
Email: \mbox{ehud.friedgut@gmail.com}}}}
\begin{document}
\date{}
\renewcommand{\thefootnote}{\fnsymbol{footnote}}
\maketitle

\begin{abstract}
Lying in order to manipulate the Gale-Shapley matching algorithm has been
studied in \cite{Dubins-Friedman} and \cite{Gale-Sotomayor-ms-machiavelli} and
was shown to be generally more appealing to the proposed-to side (denoted as
the women in \cite{Gale-Shapley}) than to the proposing side (denoted as men
there). It can also be shown that in the case of lying women, for every woman
who is better-off due to lying, there exists a man who is worse-off.

In this paper, we show that an even stronger dichotomy between the goals
of the sexes holds, namely, if no woman is worse-off then no man
is better-off, while a form of sisterhood between the lying and the
``innocent'' women also holds, namely, if none of the former are worse-off,
then neither is any of the latter.

This paper is based upon an undergraduate (``Amirim'') thesis of the
first author.

\end{abstract}

\section{Background}

\subsection{The Gale-Shapley Algorithm}

In order to standardize the notation used throughout this article, and
for sake of self-containment, let us quickly recap the scenario introduced in
\cite{Gale-Shapley} and the Gale-Shapley Algorithm introduced there:

Let $W$ and $M$ be equally-sized finite sets of women and men, respectively.
Let each member of these sets have a strict order of preference regarding
the members of the set this member does not belong to.

\begin{defn}
A one-to-one map between $W$ and $M$ is called a matching.
\end{defn}

\begin{defn}
A matching is called unstable under the given orders of preference if there
exist two matched couples $(w, m)$ and $(\tilde{w}, \tilde{m})$ such that
$w$ prefers $\tilde{m}$ over $m$ and $\tilde{m}$ prefers $w$ over $\tilde{w}$.
A matching which is not unstable is called stable.
\end{defn}

\begin{thm}[\cite{Gale-Shapley}]
The following algorithm stops and yields a stable matching between $W$
and $M$ (in particular, such a matching exists), and no stable matching
is better for any man $m \in M$:
The algorithm is divided into steps, to which we shall
refer as ``nights''. On each night, each man serenades under the window of the
woman he prefers most among all women who have not (yet) rejected him, and
then each woman, under whose window more than one man serenades, rejects every
man who serenades under her window, except for the man she prefers most
among these men. The algorithm stops on a night on which no man is rejected by
any woman, and then each woman is matched with whoever has serenaded under her
window on this night.
\end{thm}

Generalizations to this algorithm which include different sizes for $W$ and
$M$, preference lists which do not include all the members of the opposite sex
(``blacklisting''), and one-to-many matchings, are also presented in
\cite{Gale-Shapley} and will be discussed in section \ref{generalizations}.

\subsection{Previous Results}
As quoted above, the stable matching given by the Gale-Shapley algorithm
is optimal (out of all stable matchings) for each man. A somewhat reverse
claim holds regarding the women:

\begin{thm}[\cite{Roth-Sotomayor-book},
 derived from a more general theorem by Knuth]
No stable matching is worse for any woman $w \in W$ than the
matching given by the Gale-Shapley algorithm.
\end{thm}

The benefits of the Gale-Shapley matching algorithm for the men have been
demonstrated even further in \cite{Dubins-Friedman},
where it was shown that no man can get a better match by lying
about his preferences (i.e., manipulate the algorithm by declaring a false
order of preference, assuming the algorithm is run according to the ``true''
preferences of all women and all other men) and that no subset of the men
can all get better matches by lying in a coordinated fashion.

These observations lead to the analysis of the profitability of lying by women
in \cite{Gale-Sotomayor-ms-machiavelli}, where it is proven that if more that
one stable matching exists, then at least one woman can get a better match
by lying about her preferences.

It is easy to show that any woman who is better off as a result of someone's
lie (whoever that liar or those liars may be and whatever their lie may be) is
matched (due to the lie) to someone who is now worse off.
(Indeed, if they are both better off, then the original match can not be stable.)
In other words, for every woman who is better-off, some man is worse-off.

As mentioned above, in the next section we will show that an even stronger
dichotomy between the goals of the sexes holds, while a form of sisterhood
between the lying and the ``innocent'' women also holds.

\section{Monogamous Matchings}

\subsection{The Theorem}

Let $W$ and $M$ be equally-sized finite sets of women and men, respectively.
Let each member of these sets be endowed with a strict order of preference
regarding the members of the other set. These will be referred to as ``the true
preferences'', the application of the Gale-Shapley algorithm according to these
preferences will be referred to as \OA (original algorithm), and resulting
matching will be referred to as ``the original matching''.

Assume that a subset of the women, denoted $L$ (for liars) declare
false orders of preference for themselves. Call the application
of the Gale-Shapley algorithm according to these false preferences for
the members of $L$, and according to the real preferences of other
members of $W$ (referred to, henceforth, as innocent) and of all
the members of $M$, \NA, and the resulting matching --- ``the new matching''.

\begin{defn}
A person $p \in W \cup M$ is said to be ``better-off'' (resp.\ ``worse-off'') if
$p$ prefers, according to their true order of preference, their
match according to the new matching (resp.\ the original matching)
over their match according to the original matching (resp.\ the new
matching).
\end{defn}

Let us now phrase our main result for the above conditions.

\begin{thm}\label{monogamous}
Under the above conditions, if no lying woman is worse-off, then:
\begin{enumerate}[(a)]
\item \label{no-sad-woman} No woman is worse-off.
\item \label{no-happy-man} No man is better-off.
\end{enumerate}
\end{thm}

\subsection{Proofs for a Special Case}\label{special-case}

\begin{defn}
A woman $l \in L$ is said to be lying in a personally-optimal way if, all other
orders of preference being the same, there is no other false order of
preference she can declare which will result in her being even better-off,
i.e. her being matched with a man she (truly) prefers over the man matched
to her by the new matching.
\end{defn}

Theorem \ref{monogamous} is easily provable in the special case that all women
in $L$ lie in a personally-optimal way, as it is possible to show that in this
case, the new matching is stable under the true preferences, and thus, no
woman is worse-off, and no man is better-off, than under the original matching.
Nonetheless, it will be shown below that in some cases, when liars may
coordinate their lies, it is logical to lie
in a non-personally-optimal way, and that in these cases the resulting matching
might be unstable under the true preferences.

It can be noted that if we examine a ``lying game'' between the member of $L$,
where a player's strategy is a declaration of a specific order of preference
for theirself, and the utility for each player is determined by the
position of her new match on their true preference list, then in this
game, the conditions of the special case described above are met if
and only if the set
of lies constitutes a Nash equilibrium.

Roth (private communication, Dec. 2007) suggested the following sketch
of a proof to part \ref{no-sad-woman} of Theorem \ref{monogamous}:
\begin{itemize}
\item If women can do better than to state their true preferences, they can do
so by truncating their preferences.
\item Truncating preferences is the opposite of extending preferences (as
discussed in \cite{Roth-Sotomayor-book} in the context of adding new players)
\item When any woman extends her preferences, it harms the other women.
\end{itemize}

(As already noted, truncating preference lists is discussed in section
\ref{generalizations} below.) While this also proves Theorem \ref{monogamous}
in the special case discussed above (in
which all women lie in a personally-optimal way),
it appears not to prove Theorem \ref{monogamous} itself (i.e., with no
additional assumptions), as while it is true that if a woman
can do better than to state her true preferences then she can do better by
truncating them
(and even more so, there always exists a truncation of her true preferences
which constitutes a personally-optimal lie for her), it turns out that there
may exist some man whom she can secure for
herself by submitting false preferences, but not by truncating her true
preferences.
This will be illustrated by an example below.

\subsection{When a Lie Needs Not be Optimal}

As was shown 
in section \ref{special-case},
Theorem \ref{monogamous} is easily provable if each
lying woman lies in a personally-optimal way, or if the new
matching is stable under everyone's true preferences.

Before continuing to the main proof, let us give an example of a scenario
with the following properties:
\begin{enumerate}
\item No woman can do better by lying alone (while all others tell the truth).
\item In every conspiracy by more than one woman to lie so that none of them
are worse-off and at least one is better-off, there exists a woman who does not
lie in a personally-optimal way. In other words, it is logical to lie in a
non-optimal way.
\item The resulting matching is not stable under the true preferences and can
not be achieved by simply truncating women's true preference lists (even if
that is allowed).
\end{enumerate}

This example includes four women ($w_1, ..., w_4$) and four men
($m_1, ..., m_4$). The orders of preference for the women fulfill:

\begin{itemize}
\item $w_1$: First choice: $m_3$, second choice: $m_1$.
\item $w_2$: First choice: $m_3$, second choice: $m_1$.
\item $w_3$: Prefers $m_2$ over $m_1$ and prefers $m_1$ over $m_3$.
\item $w_4$: Any order of preference.
\end{itemize}

The orders of preference for the men fulfill:

\begin{center}
\begin{tabular}{ c c c c c }
\toprule
Man   & 1st choice & 2nd choice & 3rd choice & 4th choice \\
\midrule
$m_1$ & $w_1$      & $w_3$      & $w_2$      & $w_4$      \\
$m_2$ & $w_2$      & $w_3$      & any        & any        \\
$m_3$ & $w_3$      & $w_2$      & $w_1$      & $w_4$      \\
$m_4$ & $w_1$      & $w_4$      & any        & any        \\
\bottomrule
\end{tabular}
\end{center}

Let us examine \OA:

\begin{center}
\begin{tabular}{ c c c c c }
\toprule
\backslashbox{Night}{Window} & $w_1$      & $w_2$ & $w_3$ & $w_4$ \\
\midrule
1                            & $m_4, m_1$ & $m_2$ & $m_3$ &       \\
2                            & $m_1$      & $m_2$ & $m_3$ & $m_4$ \\
\bottomrule
\end{tabular}
\end{center}

It is clear from this examination that any conspiring subset of the women
wishing to make a difference
(in the scenario which does not allow truncating preference lists)
has to include $w_1$. Also, it can be verified that $w_1$ can not lie alone
(while all others tell the truth) and become better-off.
Indeed, in order to become better-off, $w_1$ has to be matched, under the new
matching, with $m_3$, and since he prefers $w_2$ over her, this entails
his rejection by $w_2$, but since he is the first choice of $w_2$, this
means $w_2$ has to lie and declare that she prefers another man over him.
Let us assume, then, that $L = \{w_1, w_2\}$. (It can be verified that the
result would not change even if we admit more women to $L$.) 

It can easily be verified that, given the true orders of preference for
everyone except $w_1$ and $w_2$, there exists exactly one combination of false
orders of preference for these two women (up to changes in their preferences
regarding men who will not reach their windows as long as everyone else tells
the truth) which will cause them both to be better-off:

\begin{itemize}
\item $w_1$ must declare she prefers $m_3$ over $m_4$ and $m_4$ over $m_1$.
\item $w_2$ must declare she prefers $m_1$ over $m_3$ and $m_3$ over $m_2$.
\end{itemize}

Let us examine \NA under these false orders of preference:

\begin{center}
\begin{tabular}{ c c c c c }
\toprule
\backslashbox{Night}{Window} & $w_1$      & $w_2$      & $w_3$      & $w_4$ \\
\midrule
1                            & $m_4, m_1$ & $m_2$      & $m_3$      &       \\
2                            & $m_4$      & $m_2$      & $m_1, m_3$ &       \\
3                            & $m_4$      & $m_2, m_3$ & $m_1$      &       \\
4                            & $m_4$      & $m_3$      & $m_1, m_2$ &       \\
5                            & $m_4$      & $m_1, m_3$ & $m_2$      &       \\
6                            & $m_3, m_4$ & $m_1$      & $m_2$      &       \\
7                            & $m_3$      & $m_1$      & $m_2$      & $m_4$ \\
\bottomrule
\end{tabular}
\end{center}

Indeed, $w_1$ and $w_2$ are both better-off (and, as Theorem \ref{monogamous}
states, no other woman is worse-off, and the even innocent $w_3$ is better-off
as well) however this lie is clearly not personally-optimal for $w_2$,
since she can declare
any order of preference under which her first choice is $m_3$, and become even
better-off by being matched with him, but if she tries to do that (either by
telling the truth, or by lying), then $w_1$ becomes worse-off and in this case
it is better for $w_1$ to tell the truth which, as stated earlier, will cause
the new matching to be identical to the old matching, and hence will cause
$w_2$ to be matched with $m_2$, ending up in a worse situation than had she
lied the above-described non-personally-optimal lie.

In the above-described lying game between $w_1$ and $w_2$, the utility for
each of the players given these strategies is higher than her utility in
any Nash equilibrium, and this is the only pair of strategies
(up to the degrees of freedom discussed above) with this property.

It should be noted that the resulting new matching is neither stable under
the original preferences (as $w_2$ and $m_3$ truly prefer each other over
their respective matches) nor can it be achieved by simply truncating the
preference lists of $w_1$ and $w_2$ (as any such truncation will result in
$w_2$ either having a blank list or a list with $m_3$ as her first choice),
giving, as promised, an example where the above proofs do not hold.

\subsection{A General Proof}

Before proving the general case of Theorem \ref{monogamous}, let us first
introduce a few notations: For each person $p \in W \cup M$, let us denote
the person from the opposite sex matched with $p$ under the original matching
(resp.\ new matching) with $O(p)$ (resp.\ $N(p)$).

\begin{defn}
A woman $w \in W$ is said to be a rejecter if she rejected $N(w)$ during \OA.
In this case, let the man who serenaded under her window on the night of \OA
on which she rejected $N(w)$, but whom she did not reject on that night (and
whom, therefore, she prefers over $N(w)$), be denoted $B(w)$.
\end{defn}

\begin{lemma}\label{didnt-reach-then-prefers-match}
If a man $m \in M$ never serenaded under the window of a woman $w \in W$ during
\NA (resp.\ \OA), then $m$ prefers $N(m)$ (resp.\ $O(m)$) over $w$.
\end{lemma}
\begin{proof}
Since $m$ approaches women according to his own order of preference, and since
he never reached $w$'s window during \NA (resp.\ \OA), then he must have
ended up being matched, under the new (resp.\ old) matching, to a woman he
prefers over $w$, and by definition, that woman is $N(m)$ (resp.\ $O(m)$).
\end{proof}

By the following Lemma, we need only prove part \ref{no-happy-man} of
Theorem \ref{monogamous}.

\begin{lemma}
If a woman $w \in W$ is worse-off, then $O(w)$ is better-off.
\end{lemma}
\begin{proof}
Since $w$ is worse-off, $w$ prefers $O(w)$ over $N(w)$ according to her true
order of preference. 
Since it is given that no liar is worse-off, then $w$ is not a liar,
and therefore she declared that she prefers $O(w)$ over $N(w)$ also during \NA.
Therefore, since $w$ prefers $N(w)$
over any other man who serenaded under her window
during \NA, it follows that $O(w)$ could not have serenaded under her window
on any night
during \NA. Thus, by Lemma \ref{didnt-reach-then-prefers-match}, $O(w)$ prefers
$N(O(w))$ over $w=O(O(w))$, making $O(w)$ better-off.
\end{proof}

\begin{lemma}\label{better-off-then-new-rejecter}
If a man $m \in M$ is better-off, then $N(m)$ is a rejecter.
\end{lemma}

\begin{proof}
Since $m$ is better-off, $m$ prefers $O(m)$ over $N(m)$.
Therefore, since $m$ approaches women according to his own order of preference,
$m$ would not serenade under $O(m)$'s window before first having been
rejected by $N(m)$.
Now, since $m$ serenades under $O(m)$'s window during the last night of \OA,
it follows that he must have been rejected by $N(m)$ on some
earlier night during \OA. Since $N(m)$ rejected $m=N(N(m))$ during
\OA, then she is a rejecter.
\end{proof}

\begin{lemma}\label{rejecter-then-worse-off}
If a woman $w \in W$ is a rejecter, then she is worse-off.
\end{lemma}

\begin{proof}
By induction and transitivity,
$w$ prefers $O(w)$ over any man she rejected
during \OA and therefore, being a rejecter, over $N(w)$, and is thus
worse-off.
\end{proof}

\begin{lemma}\label{rejecter-then-better-man-prefers-new-over-her}
If a woman $w \in W$ is a rejecter, then $B(w)$ prefers $N(B(w))$ over $w$.
\end{lemma}

\begin{proof}
By Lemma \ref{rejecter-then-worse-off}, $w$ is not a liar, and as such,
her order of preference during \NA is her true order of preference.
Since $w$ is a rejecter, $w$ prefers $B(w)$ over $N(w)$. Therefore,
as $w$ truly prefers $N(w)$ over any other man who serenaded under her window
during \NA, it follows that $B(w)$ could not have serenaded under $w$'s
window during \NA.
Thus, by Lemma \ref{didnt-reach-then-prefers-match}, $B(w)$ prefers $N(B(w))$
over $w$.
\end{proof}

\begin{lemma}\label{order-lemma}
If a woman $w \in W$ is a rejecter, then
\begin{enumerate}
\item $N(B(w))$ is a rejecter.
\item During \OA, $w$ rejected $N(w)$ on a strictly later night than the night on which $N(B(w))$ rejected $B(w)$.
\end{enumerate}
\end{lemma}

\begin{proof}
Throughout this proof, for the sake of conciseness, wherever we discuss
whether a rejection or serenading occurred, or when it occurred, we always refer
to \OA. Since $w$ is a rejecter, then
according to Lemma \ref{rejecter-then-better-man-prefers-new-over-her},
$B(w)$ prefers $N(B(w))$ over $w$. Also note that by definition of $B(w)$,
$B(w)$ serenaded under $w$'s window on some night.
By these two observations, and since $B(w)$ approaches women according to his
own order of preference,
$N(B(w))$ rejected $B(w)$, and did so on a night strictly earlier than any night
on which $B(w)$ serenaded under $w$'s window.
Since $N(B(w))$ rejected $B(w)=N(N(B(w)))$, then she is a rejecter.
Moreover, since, by definition of $B(w)$, $w$ rejected $N(w)$ on a
night on which $B(w)$ serenaded
under her window as well, then this night is a strictly later night than the
night on which $N(B(w))$ rejected $B(w)$.
\end{proof}

To complete the proof of Theorem \ref{monogamous}, assume for contradiction
that there exists a better-off man $m \in M$. Let us denote $w_1 = N(m)$, and
by Lemma \ref{better-off-then-new-rejecter}, she is a rejecter.
Now, for each $i \in \mathbb{N}$, assume by induction that
$w_i$ is a rejecter and set $w_{i+1} = N(B(w_i))$. By Lemma \ref{order-lemma},
$w_{i+1}$ is a rejecter. Moreover, by that lemma, during \OA, $w_i$ rejected $N(w_i)$ on a night strictly later than the night on
which $w_{i+1}$ rejected $B(w_i)=N(N(B(w_i)))=N(w_{i+1})$.

From finiteness of $W$, there must exist $i < j$ such that $w_i=w_j$, but by induction, since $i<j$, then during \OA,
$w_i$ rejected $N(w_i)$ on a night strictly later than the night on which $w_j$ rejected $N(w_j)$ --- a contradiction.
\qed

\begin{cor}
Under the conditions of Theorem \ref{monogamous}, every better-off woman,
is matched to a worse-off man.
\end{cor}

\section{Generalizations}
\label{generalizations}

\subsection{Polygamous Matchings}
In \cite{Gale-Shapley}, a one-to-many version of the algorithm was specified
and proposed as a way to assign students to colleges (on each algorithm-step,
the students apply to their favorite not-yet-rejected-by college, and then
each college rejects all applicants except for the most preferred ones,
according to the quota of this college) and proved that the resulting matching
is stable and that it is optimal (within all stable matches)
for each applicant. It can also be shown that it is the worst (within all
stable matches) for each college.

In \cite{Roth}, it was made public that the assignment of medical interns to
hospitals in the USA has been done using a similar algorithm since 1951,
however, in this algorithm the roles were switched and the hospitals were the
'proposers' and the resulting matching is optimal (within all stable matches)
for each hospital and worst for each intern.

In this subsection, we will prove a generalization of Theorem \ref{monogamous}
for these scenarios. In order to ease the transition from the previous section,
we will maintain the notation of women and men, and refer to the above
scenarios as ``the polygamous scenarios''. Let us redefine our notation for these
scenarios:

Let $W$ and $M$ be finite lists of women and men, respectively, and let each
person be endowed, as before, with a strict order of preference with regards
to the members of the opposite sex. For each person $p \in W \cup M$, define
$n_p$ to be the ``quota'' of this person, i.e., the amount of spouses from the
opposite sex this person seeks. As our goal is the generalization of both
algorithms described above, the reader may assume, for ease of readability,
that either all women are monogamous
($\forall w \in W: n_w = 1$) or all men are monogamous
($\forall m \in M: n_m = 1$), however the rest of this paper holds verbatim even if
this is not the case.
(If both the women and men are monogamous we are reduced to the scenario given
in the previous section.) We will also assume for now that
$\sum_{w \in W} n_w = \sum_{m \in M} n_m$. (This guarantees that when the
algorithm stops, each person $p$ is matched with exactly $n_p$ people of
the opposite sex --- this reduces in the monogamous case to $W$ and $M$
being of identical size.)

\begin{defn}
A map between W and M, mapping each $p \in W \cup M$ to
exactly $n_p$ members of the set $p$ does not belong to, is
called a matching.
\end{defn}

The definition of instability of a matching requires another detail which
was inferred from the other requirements in the monogamous scenario:

\begin{defn}\label{matching-quotas-instability}
A matching is called unstable under the given orders of preference if there
exist two matched couples $(w,m)$ and $(\tilde{w}, \tilde{m})$ such that:
\begin{enumerate}
\item $w$ is not matched with $\tilde{m}$.
\item $w$ prefers $\tilde{m}$ over $m$.
\item $\tilde{m}$ prefers $w$ over $\tilde{w}$.
\end{enumerate}
As before, a matching which is not unstable is called stable.
\end{defn}

On each night of the polygamous algorithm, each man $m \in M$ serenades under
the windows of $n_m$ women that he prefers the most out of all women who have
not (yet) rejected him, and then each woman $w \in W$, under whose
window more than $n_w$ men serenade, rejects every man who
serenaded under her window, except for the $n_w$ men she prefers most among
them.

Now, as before, assume that a subset of the women, denoted $L$, declare
false orders of preference for themselves.

For a person $p \in P$, Let us denote the set of people of the opposite sex
matched with this person under the original matching (resp.\ the new matching)
with $O(p) = \{o_1^p,...,o_{n_p}^p\}$ (resp.\ $N(p) = \{n_1^p,...,n_{n_p}^p\}$)
such that $p$ prefers $o_i^p$ over $o_{i+1}^p$ (resp.\ prefers $n_i^p$ over
$n_{i+1}^p$).

Before we phrase the polygamous version of Theorem \ref{monogamous}, we
have to redefine the circumstances under which a person is said to be better-,
or worse-off.
It should be noted that while in the monogamous scenario each person's order
of preference yields a full order on the set of possible matches for that
person (i.e., people of the opposite sex), in the polygamous case
each person's order of preference yields a partial order on the set of possible
matches for that person (i.e., $n_p$-tuples of people of the opposite sex).
This introduces an asymmetry between the following two definitions,
which did not exist in the monogamous scenario.

\begin{defn}\label{weakly-better-off}
A woman $w \in W$ is said to be weakly better-off if for each
$1 \le i \le n_w$, $w$ does not prefer $o_i^w$ over $n_i^w$.
\end{defn}

\begin{defn}\label{gain-only-worse-matches}
A man $m \in M$ is said to have gained only worse matches if he
prefers each member of $O(w)$ over each member of $N(w) \setminus O(w)$.
(Note that this condition is met in the special case in which $N(w)=O(w)$.)
\end{defn}

\begin{thm}\label{polygamous}
Under the above conditions, if all lying women are weakly better-off, then:
\begin{enumerate}[(a)]
\item All women are weakly better-off.
\item \label{only-sad-men} All men have gained only worse matches.
\end{enumerate}
\end{thm}

The special case of Theorem \ref{polygamous} in which all women are polygamous
and men are monogamous
can be easily proven by reduction to the conditions of Theorem \ref{monogamous}
by ``replicating'' each polygamous woman $w \in W$ into $n_w$ monogamous women
$\{(w,i)\}_{i=1}^{n_w}$, each having the same order of preference as $w$. For
each man $m \in M$, replace $w$ on his list of preferences with these women,
in such a way that he prefers $(w, i)$ over $(w, i+1)$ for all $i$. A reduction
along these lines was used in \cite{Dubins-Friedman} to generalize certain
properties of the monogamous scenario to the polygamous-women scenario and it
was shown there that the men matched with $(w,1),...,(w,n_w)$ by the
monogamous
algorithm are exactly those matched with $w$ by the
polygamous algorithm. Furthermore, in the notations of this paper,
it can be shown that the man matched with $(w,i)$ by \OA is $o_i^w$ and by
\NA is $n_i^w$, thus completing the reduction since this yields that
$w$ is weakly better-off if and only if none of the monogamous women
$(w,i)$ is worse-off, and since it is clear that as all men are monogamous,
if a man is not better-off under the reduction, then he has only gained
worse matches before the reduction.

Unfortunately, ``replicating'' each man in a similar same way will not yield a
proof for even the monogamous women - polygamous men scenario as easily,
for it is possible for a woman $w$
to be weakly better-off before the reduction by maintaining the same
match $m$, but to be worse-off under the reduction because her
match is e.g. $(m, 2)$ instead of $(m, 1)$. Also, every replicated
monogamous man not being better-off under the reduction does not necessarily
imply that every polygamous man has gained only worse matches before the
reduction.
To make things even worse, in the general case where both men and women may be
polygamous, running the monogamous algorithm after such a
``replication'' does not even produce the same matching
as would be produced by the polygamous algorithm,
as it may lead to situations such as a couple matched to each other
with multiplicity greater than 1.

Indeed, in order to prove Theorem \ref{polygamous} in its general form we
have to retrace our steps and revisit the inner workings of the proof
of Theorem \ref{monogamous}, rewriting it to accommodate for the
generalizations we introduced. We will now redefine some of the definitions
used in that proof, and then give its
generalized form.

\begin{defn}
A woman $w \in W$ is said to be a rejecter if she rejected any of the members
of $N(w)$ during \OA. Let us denote the set of all such rejected members
$R(w)$.
\end{defn}
\begin{defn}

A man $m \in M$ is said to be a rejectee if there exists a rejecter
$w \in N(m)$ such that $m \in R(w)$. Similarly to the proof of Theorem
\ref{monogamous}, a key role will be played by a man $B(w,m)$ for which, in
a sense, $m$ was rejected. The exact definition of this man will appear
in Lemma \ref{better-man-exists}.
\end{defn}

As before, we will begin with a Lemma by which we need only prove
part \ref{only-sad-men} of Theorem \ref{polygamous}.

\begin{lemma}
If a woman $w \in W$ is not weakly better-off, then there exists $m \in O(w)$
who has not gained only worse matches.
\end{lemma}

\begin{proof}
Since it is given that all liars are weakly better-off, if follows that
$w$ is not
a liar, and therefore her order of preference during \NA is her
true order of preference. Since $w$ is not weakly better-off, there exists
$1 \le i \le n_w$
such that $w$ prefers $o_i^w$ over $n_i^w$ and thus prefers
$o_1^w, ..., o_i^w$ over $n_i^w$. By induction and by the definition of the
polygamous algorithm, there are exactly $i-1$ men who serenaded under $w$'s
window during \NA and that she (truly) prefers over $n_i^w$
(these are $n_1^w$, ..., $n_{i-1}^w$), so by the pigeonhole principle, there
exists $1 \le j \le i$ such that $o_j^w$ did not serenade under $w$'s window
during \NA. Therefore, since $o_j^w$ approaches women according to his own
order of preference, $o_j^w$ prefers all women in $N(o_j^w)$ over $w$.
Since $|N(o_j^w)|=n_{o_j^w}=|O(o_j^w)|$, and since
$w \in O(o_j^w) \setminus N(o_j^w)$, it follows that there exists
$\tilde{w} \in N(o_j^w) \setminus O(o_j^w)$ and as stated, $o_j^w$
prefers her over $w$ and hence (as $w \in O(o_j^w)$) $o_j^w$ has not gained
only worse matches.
\end{proof}

\begin{lemma}\label{gained-non-worse-match-then-rejectee}
If a man $m \in M$ has not gained only worse matches, then he is a
rejectee.
\end{lemma}

\begin{proof}
Since $m$ has not gained only worse matches, there there exist $w
\in O(m)$ and $\tilde{w} \in N(m) \setminus O(m)$ such that $m$
prefers $\tilde{w}$ over $w$. Since $w \in O(m)$, it follows that
$m$ serenaded under $w$'s window during \OA. Since $m$ approaches
women according to his own order of preference, he would not have serenaded
under $w$'s window during \OA without having serenaded on the same night, or on
a previous night, under $\tilde{w}$'s window. However, since $m$ is
not matched with $\tilde{w}$ at the end of \OA, then $\tilde{w}$ must have
rejected him
during \OA and therefore, by definition, $m \in R(\tilde{w})$,
and thus $m$ is a rejectee.
\end{proof}

\begin{lemma}\label{rejecter-then-not-weakly-better-off}
If a woman $w \in W$ is a rejecter, then she is not weakly better-off.
\end{lemma}

\begin{proof}
Since $w$ is a rejecter, there exists $r \in R(w)$ and by
the definition of $R(w)$ there exists $1 \le i \le n_w$ such that
$n_i^w = r$. Since, by definition of $R(w)$, $w$ rejected $r$ during
\OA, and since, by induction and transitivity, $w$ prefers each member of $O(w)$
over any man she rejected during \OA, it follows that she
prefers $o_i^w$ over $r=n_i^w$, and is, thus, not weakly better-off.
\end{proof}

\begin{lemma}\label{better-man-exists}
If a woman $w \in W$ is a rejecter, then for each $r \in R(w)$ there exists
a man $B(w,r)$ such that
\begin{enumerate}
\item $B(w,r)$ serenaded under $w$'s window during \OA on the night
on which she rejected $r$, but $B(w,r)$ was not rejected by her on that night.
\item $B(w,r)$ prefers each member of $N(B(w,r))$ over $w$.
\end{enumerate}
(If more than one such man exists, define $B(w,r)$ to be one of these men,
arbitrarily.)
\end{lemma}

\begin{proof}
By Lemma \ref{rejecter-then-not-weakly-better-off}, $w$ is not a liar,
and as such, her order of preference during \NA is her true order of
preference. Let $B$ be the set of all men who serenaded under her window
on the night during \OA on which she rejected $r$, but who were not rejected
by her on that night. By definition of the algorithm, $w$ prefers
each member of $B$ over $r$, and $|B|=n_w$. Since $w$ did not reject
$r$ during \NA (since $r \in N(w)$), and since the order of preference
of $w$ during \NA is her true order of preference, then not all
member of $B$ serenaded under her window during \NA (for she is matched
under \NA to the set of the $n_w$ men that she prefers most out of all
the men who serenaded under her window during \NA). Define, therefore,
$B(w,r)$ to be a member of $B$ who did not serenade under $w$'s window
during \NA. Since $B(w,r)$ did not serenade under
$w$'s window during \NA, and since he approaches women according to his own
order of preference, it follows that he prefers each member of
$N(B(w,r))$ over $w$.
\end{proof}

\begin{lemma}\label{poly-order-lemma}
If a woman $w \in W$ is a rejecter, then for each $r \in R(w)$ there exists
a woman $\tilde{w} \in N(B(w,r))$ such that
\begin{enumerate}
\item $\tilde{w}$ is a rejecter.
\item $B(w,r) \in R(\tilde{w})$
\item During \OA, $w$ rejected $r$ on a strictly later night than the night
on which $\tilde{w}$ rejected $B(w,r)$.
\end{enumerate}
\end{lemma}

\begin{proof}
Once again, throughout this proof, for the sake of conciseness, wherever
we discuss whether a rejection or serenading occurred, or when it occurred, we
always refer to \OA.
Since $w$ is a rejecter, then by Lemma \ref{better-man-exists},
$B(w,r)$ serenaded under $w$'s window (on the
night on which she rejected $r$) and $B(w,r)$ prefers each member of
$N(B(w,r))$ over
$w$. Hence, as $B(w,r)$ approaches $n_{B(w,r)}$ women each night according to
his own order of preference, and as $|N(B(w,r))| = n_{B(w,r)}$, it follows that
$B(w,r)$ serenaded
on earlier nights under each of the windows of $N(B(w,r))$
and was rejected by at least one of them on a night strictly earlier than
any night on which he serenaded under $w$'s window --- let us denote
such a woman $\tilde{w}$. Since $\tilde{w}$ rejected $B(w,r)$ (to whom
she is matched under the new matching), it follows that she is a rejecter and
that $B(w,r) \in R(\tilde{w})$. Moreover, since, by definition of $B(w,r)$,
$w$ rejected $r$ on a night on which
$B(w,r)$ serenaded under her window as well,
then this night is a strictly later night than the night
on which $\tilde{w}$ rejected $B(w,r)$.
\end{proof}

To complete the proof of Theorem \ref{polygamous}, assume for contradiction
that there exists a man $m_1 \in M$ who has not gained only worse matches.
By Lemma \ref{gained-non-worse-match-then-rejectee}, $m_1$ is a rejectee,
therefore there exists a rejecter $w_1 \in N(m_1)$ such that $m_1 \in R(w_1)$.
Now, for each $i \in \mathbb{N}$, assume by induction that $w_i$ is a rejectee
and that $m_i \in R(w_i)$ and set, by Lemma \ref{better-man-exists},
$m_{i+1} = B(w_i, m_i)$. By Lemma \ref{poly-order-lemma} , there exists a
rejecter $w_{i+1} \in N(m_{i+1})$. Moreover, by that lemma, during \OA,
$w_i$ rejected $m_i$ on a night strictly later than the night on which
$w_{i+1}$ rejected $m_{i+1}$.

From finiteness of $W \times M$, there must exist $i < j$ such that
$w_i = w_j$ and $m_i = m_j$, but by induction, since $i<j$,
then during \OA, $w_i$ rejected $m_i$ on a night strictly later than the night
on which $w_j$ rejected $m_j$ --- a contradiction.
\qed

\subsection{Blacklists and Mismatched Quotas}

As mentioned before, in \cite{Gale-Shapley}, it is not required that
the sum of the quotas of all colleges be the same as the number of
applicants, resulting is some colleges not fulfilling their quotas or
some applicants not being accepted to any college. Moreover, it is allowed
for a college to remove some of the students from its preference list,
indicating that the college is unwilling to accept these candidates even
if it means that its quota will not be met. Similarly, it is allowed for
an applicant to remove some of the colleges from their preference lists,
indicating their unwillingness to attend these colleges even at the
risk of not being accepted to any college.

The modified algorithm-step for this scenario (in the notations
used throughout this document) is that if, by a certain night, the
number of women who have not blacklisted or (yet) rejected a man $m \in M$,
and are not
blacklisted by him, is less than $n_m$, then on that night he will
serenade under the windows of all these women. (Otherwise, if there are at
least $n_m$ such women, he will, as before, serenade under the windows of
the $n_m$ women that he prefers most out of them.)

Let us adjust our notations for this scenario, and then prove
that the result of this paper still holds under it:

\begin{defn}
A map between W and M, mapping each $p \in W \cup M$ to
{\bf at most } $n_p$ members of the set $p$ does not belong to, is
called a matching.
\end{defn}

\begin{defn}
A matching is called unstable under the given orders of preference if either
the conditions of Definition \ref{matching-quotas-instability} are met,
or if there exists a matched couple $(w,m)$ and a woman (resp.\ man)
$p$ such that:
\begin{enumerate}
\item $p$ is matched with less than $n_p$ spouses.
\item $p$ has not blacklisted $m$ (resp.\ $w$).
\item $m$ (resp.\ $w$) prefers $p$ over $w$ (resp.\ $m$).
\end{enumerate}
Once again, a matching which is not unstable is called stable.
\end{defn}

For a person $p \in P$, let us denote $O(p)$ and $N(p)$ as before, and for
each $1 \le i \le |O(p)|$ (resp.\ $1 \le i \le |N(p)|$),
let us denote $o_i^p$ (res.\ $n_i^p$) as before as well.

\begin{defn}
A woman $w \in W$ is said to be weakly better-off if the following conditions
hold:
\begin{enumerate}
\item $N(w)$ contains none of the men blacklisted by $w$.
\item $|O(w)| \le |N(w)|$
\item For each $1 \le i \le O(w)$, $w$ does not prefer $o_i^w$ over
$n_i^w$.
\end{enumerate}
If, in addition, either for some $1 \le i \le O(w)$, $w$ prefers $n_i^w$
over $o_i^w$ (``$w$ has improved her matches''), or $|O(w)| < |N(w)|$
(``$w$ has gained matches''), then $w$ is said to be better-off.
\end{defn}

The definition of a man having gained only worse matches remains
unchanged. Specifically, it should be emphasized that this definition does
not require that $|N(m)| \le |O(m)|$, but we will show in Corollary
\ref{unchanged-match-sizes} that under the above conditions the inverse
is an impossibility. Similarly, the same corollary will show that 
under the above conditions, a woman may not be
better-off due to gaining matches, but only due to improving them.
It should also be noted that, by the definition of the
algorithm, it is not possible for any person declaring their true
preferences to be matched, under the new matching, to any person blacklisted
by them.

\begin{cor}\label{mismatched-quotas-and-blacklists}
Theorem \ref{polygamous} holds under the above conditions.
\end{cor}

\begin{proof}
Corollary \ref{mismatched-quotas-and-blacklists} can be proven by
reducing to the conditions of Theorem \ref{polygamous} by introducing,
for each person $p$, $n_p$ new monogamous people of the opposite sex
$\emptyset_1^p,...,\emptyset_{n_p}^p$, each of whom prefers $p$ the most
(the rest of their orders of preference can be arbitrary).
The original (resp.\ new) order of preference for $p$ will now be:
first, all of the original people of the opposite sex who are not
originally (resp.\ newly) blacklisted by $p$, ordered according to her
or his given original (resp.\ new) order of preference,
then $\emptyset_1^p,...,\emptyset_{n_p}^p$ in this order, and then,
in an arbitrary order, the original (resp.\ new) blacklist of $p$ and the
people newly-introduced for the other people of the same sex as $p$.
It is straight-forward to verify that the quotas of all women (original and
newly-introduced) match the quotas of all men (original and newly-introduced).
It is left for the reader to verify that the old (resp.\ new) matching
before the reduction is identical to the old (resp.\ new) matching under
the reduction, after the newly-introduced people have been removed from it,
that each of the original women is weakly better-off before the reduction
iff she is weakly better-off under the reduction, and that if each of the
original men has gained only worse matches under the reduction, then he has
also gained only worse matches before the reduction.
\end{proof}

\begin{cor}\label{unchanged-match-sizes}
Under the conditions of Corollary \ref{mismatched-quotas-and-blacklists},
for each person $p \in W \cup M$, it holds that $|N(p)| = |O(p)|$.
\end{cor}

\begin{proof}
Observe that for each person $p \in W \cup M$, their old
(resp.\ new) matches under the reduction are their old (resp.\ new)
matches before the reduction, ``padded'' by as many people as needed from the
top of $\emptyset_1^p,...,\emptyset_{n_p}^p$ to fulfill their quota of
$n_p$ matches. Assume, for contradiction, that there exists a man $m \in M$
such that (before the reduction) $|N(m)|>|O(m)|$. Let $o=|O(m)|$ (before
the reduction), then under the reduction, $m$ is matched with
$\emptyset_{n_p-o}^p$ under the original matching, but not under the new
matching. Recall that $\emptyset_{n_p-o}$ prefers $p$ the most, so she
is not weakly better-off --- a contradiction.
We have thus shown that for each man $m \in M$, $|N(m)| \le |O(m)|$,
and since all women are weakly better-off, then for each woman $w \in W$,
$|N(w)| \ge |O(w)|$, thus we have
\[
\sum_{w \in W}N(w) \ge \sum_{w \in W}O(w) = \sum_{m \in M}O(m) \ge
\sum_{m \in M}N(m) = \sum_{w \in W}N(w)
\]
which yields that all members of this inequality are actually equal.
Since for every woman $w \in W$, $|N(w)| \ge |O(w)|$, (res.\ for
every man $m \in M$, $|N(m)| \le |O(m)|$,) and since summing
over all women (resp.\ all men) we get an equality,
it follows that the equality holds for each woman (resp.\ each man) separately.
\end{proof}

It should be noted that in \cite{Gale-Sotomayor-remarks}, Corollary
\ref{unchanged-match-sizes} is proved using an entirely different approach
for the special case of the monogamous scenario where the new match is stable,
and in \cite{Roth-Sotomayor-book}, it is proved also for the polygamous
scenario where the new match is stable. Corollary \ref{unchanged-match-sizes}
generalizes these results, as it does not require stability of the new match.

\begin{cor}
Under the conditions of Theorem \ref{monogamous}, if $|L|=1$ and the
lying woman is better-off, then so is some innocent woman too.
\end{cor}

\begin{proof}
Since the liar is better-off, then she is newly matched with some man
to whom she was not originally matched, and since the number of matches
for this man remained the same, then he is no longer matched with some
woman to whom he was originally matched, and hence the set of matches
for that woman has changed, and thus, since she is weakly better-off, she is
better-off.
\end{proof}

\section{Summary and Open Problems}
We have shown throughout this paper that even in very general scenarios,
a form of sisterhood exists in the Gale-Shapley matching algorithm, both
in terms of not harming each other, and in terms of not helping any man.

If one were to define the concepts of ``weakly worse-off'' and ``having gained
only better matches'' along the lines of definitions \ref{weakly-better-off} and
\ref{gain-only-worse-matches}, then one would find that gaining only worse
matches is stronger than (i.e. a special case of)
being weakly worse-off and gaining only better
matches is stronger case of being better-off. Thus,
it follows that in the scenarios discussed in this paper, in a sense,
the minimum possible damage to a man is greater than the minimum possible
gain for a woman.

It is interesting to check whether, under various
conditions, the overall damage to men is, in any sense, usually
greater that the overall gain for women (and thus resulting in damage to the
entire population as a whole, which contrasts the sisterhood which exists
amongst the women).

The example given in this paper (and other examples given in the first
author's undergraduate thesis, upon which this paper is based) had to
be crafted very delicately. It is interesting to check whether,
when the orders of preference are determined by a random model, lies
can be very beneficial (either for the lying women, and for their innocent
colleagues), or whether the utility gained from such a lie is usually
relatively small. Furthermore, in many of these examples, and in many
examples specified in the literature, the preferences of every person
differ greatly from those of each of their colleagues. In the
real-world scenario of colleges and applicants, however, it is reasonable
to expect the preferences of most applicants to be similar, and the same
applies to the preferences of most colleges. There will always be differences,
but it seems that they are likely to be local, such as a permutation on
colleges which follow each other in the order of preference. It seems
unlikely, for example, for a certain college to be highly rated by half of
the applicants and poorly raged by the rest. (All this applies to the
preferences of colleges as well, naturally.) It is interesting to try and
build a probabilistic model along these lines, that matches the observed
behaviors by medical interns and colleges throughout the years, and to check
the questions raised above under such a model. In addition, since in the
real world the knowledge of each player is not full, it can be interesting
to check whether, given only knowledge of a statistical model for the
preferences of the rest of the players (and possibly full knowledge of
the preferences of just a few players), there exist strategies (for
individual players, or for small sets of players) which are expected to be
better for these players than telling the truth.
It can also be interesting to check what happens if such strategies are
simultaneously applied by more than one conspiring set of players.
\vspace{5mm}\newline\noindent
{\bf \large Acknowledgments: } We would like to thank Sergiu Hart for
useful conversations and for providing us with references to the literature.

\end{document}